%% file: main-plain.tex
\documentclass{article}

\usepackage{mathtools}
\usepackage{tikz}
\usepackage{a4wide}
\usepackage{amsthm}
\usepackage{amssymb}
\usepackage{amsmath}
\usetikzlibrary{arrows, shapes.misc}

\newtheorem{theorem}{Theorem}
\newtheorem{lemma}[theorem]{Lemma}
\newtheorem{definition}[theorem]{Definition}

\tikzset{cross/.style={cross out, draw=black, fill=none, minimum size=2*(#1-\pgflinewidth), inner sep=0pt, outer sep=0pt}, cross/.default={3pt}}
\pgfdeclarepatternformonly{very loose dots}{\pgfqpoint{-1pt}{-1pt}}{\pgfqpoint{15pt}{15pt}}{\pgfqpoint{15pt}{15pt}}%
{
    \pgfpathcircle{\pgfqpoint{0pt}{0pt}}{.5pt}
    \pgfusepath{fill}
}
\pgfarrowsdeclarecombine*{spaced stealth'}{spaced
stealth'}{stealth'}{stealth'}{space}{space}
\tikzset{>=spaced stealth'}

\newcommand{\bin}{\mathrm{bin}}

\title{On Integer Programming, Discrepancy, and Convolution\footnote{
This research was supported by German Research Foundation (DFG)
projects JA 612/20-1 and JA 612/16-1.}}
\author{Klaus Jansen \and Lars Rohwedder}

\begin{document}

\maketitle

\begin{abstract}
Integer programs with m constraints are solvable in pseudo-polynomial time in $\Delta$, the largest coefficient in a
constraint, when m is a fixed constant. We give a new algorithm with a running time of $O(\sqrt{m}\Delta)^{2m} + O(nm)$,
which improves on the state-of-the-art. Moreover, we show that improving on our algorithm for any $m$ is
equivalent to improving over the quadratic time algorithm for $(\min,~+)$-convolution. This is a strong evidence that our algorithm’s running time is the best possible. We also present a specialized algorithm with running time $O(\sqrt{m} \Delta)^{(1 + o(1))m} + O(nm)$ for testing feasibility of an integer program and also give a tight lower bound, which is based on the SETH in this case.
\end{abstract}%

\section{Introduction}
  Vectors $v^{(1)},\dotsc,v^{(n)}\in\mathbb R^m$
  that sum up to $0$ can be seen as a circle in $\mathbb R^m$ that walks from $0$ to $v^{(1)}$ to $v^{(1)} + v^{(2)}$,
  etc. until it reaches $v^{(1)} + \dotsc + v^{(n)} = 0$ again. The Steinitz Lemma~\cite{SteinitzLemma} says
  that if each of the vectors is small with respect to some norm, we can reorder them
  in a way that each point in the circle is not far away from $0$ with respect to the same norm.

  Recently Eisenbrand and Weismantel found a beautiful application of this lemma
  in the area of integer programming~\cite{DBLP:journals/talg/EisenbrandW20}.
  They looked at ILPs in standard form 
  \begin{equation}\label{eq:standard}
      \max\{c^T x : A x = b, x\in \mathbb Z^n_{\ge 0}\},
  \end{equation}
  where $A\in \mathbb Z^{m\times n}, b\in \mathbb Z^m$ and $c\in \mathbb Z^n$
  and obtained a pseudo-polynomial algorithm in $\Delta$, the biggest
  absolute value of an entry in $A$, when $m$ is treated as a constant.
  The running time they achieve is $n\cdot O(m\Delta)^{2m}\cdot \lVert b \rVert_1^2$
  for finding the optimal solution
  and $n\cdot O(m\Delta)^{m}\cdot \lVert b \rVert_1$ for finding only a feasible solution.
  This improves on a classic algorithm by Papadimitriou~\cite{DBLP:journals/jacm/Papadimitriou81}, which has a running time
  of
  \begin{equation*}
    O(n^{2m+2}\cdot (m \cdot \max\{\Delta, \lVert b \rVert_\infty\})^{(m+1)(2m + 1)}) .
  \end{equation*}
  The central idea in~\cite{DBLP:journals/talg/EisenbrandW20} is that a solution $x^*$ for the ILP above can be viewed as a walk in $\mathbb Z^m$
  starting at $0$ and ending at $b$. Every step is a column of the matrix $A$:
  For every $i\in\{1,\dotsc,n\}$ we step $x^*_i$ times in the direction of $A_i$ (see left picture in Figure~\ref{fig-steinitz}).
  By applying the Steinitz Lemma they show that there is an ordering of these steps such
  that the walk never strays off far from the direct line between $0$ and $b$ (see right picture in Figure~\ref{fig-steinitz}).
  They construct a directed graph with one vertex for every integer point near the line between $0$ and $b$
  and create an edge from $u$ to $v$, if $v - u$ is a column in $A$. The weight of the edge is
  the same as the $c$-value of the column. An optimal solution to the ILP can now be obtained by
  finding a longest path from $0$ to $b$. This can be done in the mentioned time,
  if one is careful with cycles.

  In this work we present a different algorithm for the same problem. In our approach we do not
  reduce to a longest path problem, but rather solve the ILP in a divide and conquer fashion. We
  use the (weaker) assumption that a walk from $0$ to $b$ visits a vector $b'$
  near $b/2$ at some point. The distance of this point to $b/2$ is closely related to the discrepancy
  of the matrix $A$, see Lemma~\ref{lem-split} for more details. A natural approach is to guess the vector $b'$
  and solve the problem with $Ax = b'$ and $Ax = b - b'$ independently.
  Both results can be merged to a solution for $Ax = b$. 
  In the subproblems the norm of $b$ and the norm of the solution are roughly divided in half.
  We use this idea in a dynamic program and speed up the process of merging solutions using algorithms
  for convolution. This approach leads to better running times for both the problem of finding optimal
solutions and for finding any feasible solution. We complement our study by giving (almost) tight
conditional lower bounds on the running time in which such ILPs can be solved. Finally, we discuss
some applications to Knapsack, Change Making, and Scheduling problems.

We proceed by giving a detailed outline of the results.

   \begin{figure}
\centering
\input{figure1}
  \caption{Steinitz Lemma in Integer Programming}
  \label{fig-steinitz}
   \end{figure}

   \paragraph{Optimal solutions for ILPs.}
   We show that a solution to (\ref{eq:standard}) can be found in time
   \begin{equation}\label{eq:opt}
     O(\sqrt m \Delta)^{2m} + O(nm) .
   \end{equation}
   We note that throughout the article we work with the assumption that arithmetics on the input
numbers require constant time.
   Comparing to the state-of-the-art, we remove the dependence on $b$ from the running time
   and save a factor of $n$ without increasing the dependence on $\Delta$ and even mildly improving
   the dependence on $m$.
   The running time can be improved if there exists a truly sub-quadratic algorithm for (min,~+)-convolution (see Section~\ref{sec-convolution} for details on the problem).
   However, it has been conjectured that no such algorithm exists and this conjecture is
   the base of several lower bounds in fine-grained complexity~\cite{DBLP:conf/icalp/CyganMWW17, DBLP:conf/icalp/KunnemannPS17, DBLP:conf/isit/LaberRC14, DBLP:conf/soda/BackursIS17}.
   We show that for every $m$ the running time above is essentially the best possible unless
   the (min,~+)-convolution conjecture is false.
   More formally, for every $m$ there exists no algorithm that solves
   ILP in time
   $f(m) \cdot (n^{2-\delta} + (\Delta + \lVert b \rVert_\infty)^{2m - \delta})$,
   where  $\delta > 0$ and $f$ is an arbitrary computable,
   unless there exists a truly sub-quadratic algorithm for (min,~+)-convolution.
   Indeed, this means there is an equivalence between improving algorithms for (min,~+)-convolution
   and for ILPs with fixed number of constraints.
   It may be surprising that the lower bound has a dependence on $\Delta + \lVert b \rVert_\infty$
   and the upper bound only on $\Delta$. This implies that hardness cannot come from only letting $b$ grow
   and, in particular, it rules out improvements by adding a dependence on $\lVert b \rVert_\infty$.
  Our lower bound does leave 
  open some other trade-offs between $n$ and $O(\sqrt{m}\Delta)^m$
  such as $n \cdot O(\sqrt{m} \Delta)^m$, which would be
  an interesting improvement for sparse instances, i.e., when $n \ll (2 \Delta + 1)^m$.
  Such an improvement has recently been made for Unbounded Knapsack~\cite{DBLP:journals/jcss/ChanH22}, a notable special case of $m = 1$, see also Definition~\ref{def:unbounded-knapsack}.
  A running time of $n^{f(m)} \cdot (\Delta + \lVert b \rVert_\infty)^{m-\delta}$, however, is not
  possible (see feasibility below).
   \paragraph{Feasibility of ILPs.} Finding only a feasible solution of an ILP is easier than finding an
optimal solution. It can be done in time
   \begin{equation}\label{eq:feas}
   O(\sqrt m \Delta)^{(1 + o(1))m} + O(nm)
   \end{equation}
   by solving
   a Boolean convolution problem that has a more efficient algorithm than the (min,~+)-convolution
   problem that arises in the optimization version.
   Under the Strong Exponential Time Hypothesis (SETH) this running time
   is tight except for sub-polynomial factors.
   The SETH and the Exponential Time Hypothesis (ETH) are conjectures commonly used to prove conditional
   lower bounds. The SETH asserts that the satisfiability problem (SAT) cannot be solved in time $O(2^{\delta n})$ for any $\delta < 1$,
   while the (weaker) ETH asserts that this holds for some $\delta > 0$.
   If the SETH holds, then there is no
   $n^{f(m)} \cdot (\Delta + \lVert b \rVert_\infty)^{m - \delta}$ time algorithm for testing feasibility of
   ILPs for any $\delta > 0$ and any computable function~$f$.
   \paragraph*{Comparison to previous version.} A preliminary version of this article has appeared in the
   proceedings of ITCS 2019~\cite{DBLP:conf/innovations/JansenR19}. The analysis in that version has relied completely on the Steinitz Lemma, whereas the present article uses bounds on hereditary discrepancy,
   which is a cleaner fit given the requirements in the proof. Furthermore, this change
   leads to a slightly improved base $O(\sqrt{m}\Delta)$ in the running times instead of the previous base $O(m\Delta)$. This can be improved further in case of constraint matrices with small hereditary discrepancy.
   Moreover, by utilizing specialized algorithms for linear programming
   in fixed dimension we avoid the logarithmic dependency on $\lVert b \rVert_\infty$ as in the previous version.
   This also allows us to simplify the proof by removing a lemma that bounds the norm of the solution, which
   was required earlier.
   To the applications, we added the Coin Change problem.
  \subsection*{Other related work} It is notable that the case
  where the number of variables $n$ is fixed
  and not $m$ as here behaves differently.
  There is a $2^{O(n\log(n))} \cdot |I|^{O(1)}$ time algorithm ($|I|$ being the encoding length of the input), whereas
  an algorithm of the kind $f(m) \cdot |I|^{O(1)}$ (or even $|I|^{f(m)}$) is impossible for any computable function $f$, unless $\mathrm P = \mathrm{NP}$.
  This can be seen with a trivial reduction from Unbounded Knapsack (where $m=1$).
  The $2^{O(n\log(n))} \cdot |I|^{O(1)}$ time algorithm is due to
  Kannan~\cite{DBLP:journals/mor/Kannan87} improving over a $2^{O(n^2)}\cdot |I|^{O(1)}$ time algorithm
  by Lenstra~\cite{DBLP:journals/mor/Lenstra83}. It is a long
  open question whether $2^{O(n)}\cdot |I|^{O(1)}$ is possible instead;
  see also~\cite{DBLP:journals/algorithmica/Dadush14, dadush2012integer} for progress towards this question.

  Another intriguing question is whether a similar running time
  as in this work, e.g., $(\sqrt m \Delta)^{O(m)} \cdot n^{O(1)}$, is possible when upper bounds on variables are added to the ILP and they are not counted in $m$.
  In~\cite{DBLP:journals/talg/EisenbrandW20} an algorithm for this extension is given,
  but the exponent of $\Delta$ is $O(m^2)$.

  As for other lower bounds on pseudo-polynomial algorithms for integer programming,
  Fomin et al.~\cite{DBLP:conf/esa/FominPR018} prove that the running time cannot be
  $n^{o(m / \log(m))} \cdot \lVert b \rVert_\infty^{o(m)}$ unless the
  ETH (a weaker conjecture than the SETH) fails.
  Their reduction implies that there is no algorithm with running time
  $n^{o(m/\log(m))} \cdot (\Delta + \lVert b \rVert_\infty)^{o(m)}$, since
  in their construction the matrix $A$ is non-negative and therefore columns with entries larger
  than $\lVert b \rVert_\infty$ can be discarded; thus leading to $\Delta \le \lVert b \rVert_\infty$.
  Very recently, Knop et al.~\cite{DBLP:conf/stacs/KnopPW19} show that under the ETH there is also no
  $2^{o(m\log(m))} \cdot (\Delta + \lVert b \rVert_\infty)^{o(m)}$ time algorithm.
  An interesting aspect of this function is that it matches the dependency in $m$ achieved
  here and in~\cite{DBLP:journals/talg/EisenbrandW20} up to a constant in the exponent.
  Our lower bound differs substantially from the two above.
  We concentrate on the dependency on $\Delta$ and give a precise value
  of the constant in its exponent.
  
  Linear programming in fixed dimension, that is, solving (\ref{eq:standard})
  where $x\in \mathbb R^n_{\ge 0}$ instead of $x \in \mathbb Z^n_{\ge 0}$,
  has also been studied extensively. In a seminal work~\cite{DBLP:journals/jacm/Megiddo84}, Megiddo gave
  the first linear time algorithm for $m = O(1)$. Since then there have been numerous improvements~\cite{DBLP:journals/algorithmica/MatousekSW96, DBLP:journals/ipl/Clarkson86, DBLP:journals/siamcomp/Dyer86, DBLP:journals/mp/DyerF89, DBLP:journals/jacm/Clarkson95, DBLP:journals/dcg/Seidel91, DBLP:conf/stoc/Kalai92, agarwal1993efficient, DBLP:journals/jal/ChazelleM96, DBLP:journals/siamcomp/BronnimannCM99, DBLP:journals/talg/CHAN18}.
  The currently best randomized algorithm has a running time of $m^2 n + 2^{O(\sqrt{m \log(m)})}$
  (a combination of \cite{DBLP:journals/jacm/Clarkson95, DBLP:conf/stoc/Kalai92, DBLP:journals/algorithmica/MatousekSW96}) and the best deterministic algorithm~\cite{DBLP:journals/talg/CHAN18} has a running time
  of $O(m)^{m/2} \cdot \log^{3m}(m) \cdot n$. These works typically solve the dual of this problem,
  which is equivalent by standard complementary slackness arguments. Our algorithm for ILP uses these
  results as a subroutine.
  \section{Preliminaries}
   In the remainder of the article we will assume that $A$ has no duplicate columns.
   Note that we can completely ignore a column $i$, if there is another
   identical column $i'$ with $c_{i'} \ge c_i$.
   This implies that in time $O(n m) + O(\Delta)^m$ we can reduce to an
   instance without duplicate columns and,
   in particular, with $n\le (2\Delta + 1)^m$.
   The running time can be achieved as follows.
   We create a new matrix for the ILP
   with all $(2\Delta+1)^m$ possible columns (in lexicographic order)
   and objective value $c_i = -\infty$ for all columns $i$.
   Now we iterate over all $n$ old columns
   and compute in time $O(m)$ the index of the new column corresponding to the same entries.
   We then replace its objective value with the current one if this is bigger.
   In the upcoming running times we will omit the additive term $O(nm)$ and assume
   the duplicates are already eliminated ($O(\Delta)^m$ is always dominated by actual algorithms
   running time).

  Eisenbrand and Weismantel observed that using the Steinitz Lemma (with $\ell_\infty$ norm) one can solve integer programs
  efficiently, if all entries of the matrix are small integers and the
  number of constraints is fixed.
  \begin{theorem}[Steinitz Lemma]
    Let $\lVert \cdot \rVert$ be a norm in $\mathbb R^m$
    and $v^{(1)},\dotsc,v^{(t)}\in\mathbb R^m$ such that $\lVert v^{(i)} \rVert \le \Delta$ for all $i$
    and $v^{(1)} + \cdots + v^{(t)} = 0$. Then there exists a permutation $\pi\in S_t$ such that
    for all $j\in\{1,\dotsc, t\}$
    \begin{equation*}
      \lVert \sum_{i=1}^j v^{(\pi(i))} \rVert \le m \Delta .
    \end{equation*}
  \end{theorem}
  The proof of the bound $m\Delta$ is due to Sevastyanov~\cite{sevast-steinitz} (see
  also~\cite{DBLP:journals/talg/EisenbrandW20} for a good overview).
  Our algorithmic results rely on a similar, but weaker property.
  Roughly speaking, we only need that there is some $j \approx t/2$ with $\lVert \sum_{i=1}^j v^{(\pi(i))} \rVert$ being small.
  All other partial sums are insignificant.
  As it is a weaker property, we can hope for better bounds than $m\Delta$,
  which is indeed true. The bounds we need come from discrepancy
  theory, for which we now state relevant definitions and results.
  \begin{definition}
    For a matrix $A\in\mathbb R^{m\times n}$ its discrepancy is
    \begin{equation*}
      \mathrm{disc}(A) = \min_{z\in\{0,1\}^n} \left\lVert A\left(z - \left(\frac 1 2,\dotsc,\frac 1 2\right)^T\right)\right\rVert_\infty .
    \end{equation*}
  \end{definition}
  Discrepancy theory originates in the problem of coloring the elements of a ground set with two colors
  such that a given family of subsets are all colored evenly, i.e., the number of elements of
  each color is approximately the same. When $A$ is the incidence matrix of this family of sets,
  $z$ in the definition above gives a coloring and the $\ell_\infty$ norm its discrepancy.
  Discrepancy, however, is also studied
  for arbitrary matrices. If $A$ is the matrix of a linear program as in our case, this definition
  corresponds to finding an integral solution that approximates $x=(1/2,\dotsc,1/2)^T$.
  Our algorithm is based on dividing a solution into two similar parts. Therefore,
  discrepancy is a natural measure. However, we need a definition that
  is stable when restricting to a subset of the columns.
  \begin{definition}
    The hereditary discrepancy of a matrix $A\in\mathbb R^{m\times n}$ is
    \begin{equation*}
      \mathrm{herdisc}(A) = \max_{I\subseteq\{1,\dotsc,n\}} \mathrm{disc}(A_I),
    \end{equation*}
    where $A_I$ denotes the matrix $A$ restricted to the columns $I$.
  \end{definition}
  Hereditary discrepancy is often used in the context of rounding non-integral solutions,
  see for example~\cite{DBLP:journals/ejc/LovaszSV86}.
  For our algorithm we need to split a solution $x$ into two similar parts, which can be
  seen as rounding $x/2$. The following lemma shows that
  by paying a factor of $2$ in the discrepancy we can
  also get a balanced split of the $\ell_1$ norm of the solutions.
  \begin{lemma}\label{lem-split}
    Let $x\in\mathbb Z^n_{\ge 0}$. Then there is a vector $z\in\mathbb Z^n_{\ge 0}$
    with $z_i \le x_i$ for all $i$ and
\begin{equation*}
  \left\lVert A \left(z - \frac x 2 \right) \right\rVert_\infty \le \mathrm{herdisc}(A) .
\end{equation*}
    Furthermore, if $\lVert x \rVert_1 > 1$, then there is a vector $z'\in\mathbb Z^n_{\ge 0}$
    with $z'_i \le x_i$ for all $i$,
  $1 / 6 \cdot \lVert x \rVert_1 \le \lVert z' \rVert_1 \le 5 / 6 \cdot \lVert x \rVert_1$, and
\begin{equation*}
  \left\lVert A \left(z' - \frac x 2\right) \right\rVert_\infty \le 2 \cdot \mathrm{herdisc}(A) .
\end{equation*}
  \end{lemma}
We emphasize that the lemma is symmetric in the sense that the same properties hold when substituting $z$ for $x-z$ ($z'$ for $x - z'$).
For completeness a proof of the lemma is given in the appendix.
Our algorithm's running time will depend on $\mathrm{herdisc}(A)$,
so we will state some standard bounds on it.
  \begin{theorem}[Spencer's Six Standard Deviations Suffice~\cite{spencer1985six}]\label{th:spencer}
    For every matrix $A\in\mathbb R^{m\times n}$ with biggest absolute value of an entry $\Delta$,
    \begin{equation*}
      \mathrm{herdisc}(A) \le 6 \sqrt{m} \cdot \Delta .
    \end{equation*}
  \end{theorem}
  This slightly differs from the original statement.
  The original paper considers square matrices ($n=m$) with biggest absolute value $1$
  and gives a bound of $6 \sqrt{n} = 6 \sqrt{m}$.
  However, the proof easily holds also for $6\sqrt{m}$ in non-square matrices, 
  as mentioned for example in~\cite{DBLP:journals/ejc/LovaszSV86}.
  By scaling both sides we obtain $6 \sqrt{m}\cdot\Delta$ for matrices with
  biggest absolute value $\Delta$.
  
  Spencer's proof is not constructive, that is, it is unclear how to compute the $z$ from
  the defintion of discrepancy.
  There has been significant work towards making it
  constructive~\cite{DBLP:conf/focs/Bansal10, DBLP:journals/siamcomp/LovettM15}.
  For our algorithm, however, we will not need a constructive variant.

  There are matrices for which Spencer's bound is tight up to a constant factor.
  For specific matrices it might be lower.
  The linear dependency on $\Delta$, however, is required for any matrix $A$.
  \begin{lemma}\label{disc-lb}
    For every matrix $A\in\mathbb R^{m\times n}$ with absolute value of an entry $\le \Delta$,
    \begin{equation*}
     \mathrm{herdisc}(A) \ge \frac{\Delta}{2} .
    \end{equation*}
  \end{lemma}
  This can be seen by taking $I = \{i\}$ in the definition of $\mathrm{herdisc}(A)$ with
  $A_i$ being a column with an entry of absolute value $\Delta$.
  An example where the dependency on $m$ is lower than in Spencer's
  theorem are matrices with a small $\ell_1$ norm in every column.
  \begin{theorem}[Beck, Fiala~\cite{DBLP:journals/dam/BeckF81}]\label{beck-fiala}
    For every matrix $A\in\mathbb R^{m\times n}$, where the $\ell_1$ norm of each column is at most $t$ it holds that $\mathrm{herdisc}(A) < t$.
  \end{theorem}

  \section{Algorithm}
    First, we will show how to compute the best solution $x^*$ to (\ref{eq:standard})
    with the additional constraint $\lVert x^* \rVert_1 \le K$.
    Here the running time has a logarithmic dependence on $K$.
    Then, we will remove this dependence while allowing arbitrarily
    large solutions. Further, we will elaborate an improvement
    for finding any feasible solution and show how to cope
    with unbounded problems. Finally, we give a more fine-grained
    study of the problem when the maximum entries of the rows differ.
    
    \subsection{Dynamic program}
    Let $H \ge \mathrm{herdisc}(A)$ be a given upper bound on the hereditary discrepancy of $A$.
    For every $i=0,1\dotsc,\ell = \lceil \log_{6/5}(K) \rceil$ and every $b'$ with
    $\lVert b' - 2^{i-\ell} \cdot b \rVert_\infty \le 4 H$
    we solve
    \begin{equation}
 \max\left\{c^T x : A x = b', \lVert x \rVert_1 \le \left(\frac 6 5 \right)^{i}, x\in\mathbb Z^n_{\ge 0}\right\} .\label{dyn-ilp}
    \end{equation}
    We iteratively derive solutions
    for $i$ using pairs of solutions for $i-1$. Ultimately, we will compute a solution
    for $i = \ell$ and $b' = b$.

    If $i = 0$ the solutions are trivial, since $\lVert x \rVert_1 \le 1$.
    This means they correspond exactly to the columns of $A$.
    Fix some $i > 0$ and $b'$ and let $x^*$ be an
    optimal solution to (\ref{dyn-ilp}).
    By Lemma~\ref{lem-split} there exists a $0 \le z \le x^*$ with
    $\lVert A z - b'/2 \rVert_\infty \le 2 H$ and
    \begin{equation*}
      \lVert z \rVert_1 \le 
      \begin{cases}
      \frac 5 6 \lVert x^* \rVert_1 \le \frac 5 6 \cdot \left(\frac 6 5\right)^i = \left(\frac 6 5\right)^{i-1} &\text{ if } \lVert x^* \rVert_1 > 1, \\
      \lVert x^* \rVert_1 \le 1 \le \left(\frac 6 5\right)^{i-1} &\text{ otherwise.}
      \end{cases}
    \end{equation*}
    The same holds for $x^* - z$.
    Then $z$ is an optimal solution to
    \begin{equation*}
    \max \left\{c^Tx : Ax = b'', \lVert x \rVert_1 \le \left(\frac 6 5 \right)^{i - 1}, x\in \mathbb Z^n_{\ge 0} \right\},
    \end{equation*}
    where $b'' = Az$. This is because if there was a solution $z^*$ of higher value, then $z^* + x^* - z$ would
    be feasible for~(\ref{dyn-ilp}) and have a higher value than $x^*$, contradicting its optimality.
    Likewise, $x^* - z$ is an optimal solution to
    \begin{equation*}
        \max \left\{c^T x : Ax = b' - b'', \lVert x \rVert_1 \le \left(\frac 6 5\right)^{i - 1}, x\in \mathbb Z^n_{\ge 0} \right\}.
    \end{equation*}
    We will prove that $\lVert b'' - 2^{(i-1)-\ell} \cdot b\rVert_\infty \le 4H$
    and $\lVert (b' - b'') - 2^{(i-1)-\ell} \cdot b\rVert_\infty \le 4H$.
    This implies that we can look up solutions for $b''$ and $b' - b''$ in the dynamic table
    and their sum is a solution for $b'$. Clearly it is also optimal. We do not know $b''$, but
    we can guess it: There are only $(8H + 1)^m$ candidates.
    To compute an entry, we therefore enumerate all possible $b''$ and take the two partial solutions
    (for $b''$ and $b' - b''$), where the sum of both values is maximized.
    To verify that the inequalities above holds, we calculate
    \begin{align*}
      \left\lVert b'' - 2^{(i-1)-\ell} b \right\rVert_\infty &= \left\lVert A z - \frac 1 2 b' + \frac 1 2 b' - 2^{(i-1)-\ell} b \right\rVert_\infty \\
  &\le \left\lVert A z - \frac 1 2 b'\right\rVert_\infty + \left\lVert\frac 1 2 b' - 2^{(i-1)-\ell} b \right\rVert_\infty \\
  &\le 2\cdot\mathrm{herdisc}(A) + \frac 1 2 \left\lVert b' - 2^{i-\ell} b \right\rVert_\infty
  \le 4H .
    \end{align*}
    The same holds for $b' - b''$, since $z$ and $x^* - z$ are interchangeable.
      The dynamic table has $O(H)^m \cdot \log(K)$ entries.
      To compute an entry, $O(n\cdot m) \le O(\Delta)^m \le O(H)^m$ operations are necessary during initialization and $O(H)^m$ in the iterative calculations.
      This gives a total running time of
      \begin{equation}\label{eq:rt-naive}
         O(H)^{2m} \cdot \log(K) .
      \end{equation}
    \subsection{Convolution}\label{sec-convolution}
      The careful reader may wonder, whether the computation
      of entries in the dynamic table can be improved.
      Let $D_i$ be the set of vectors $b'$ with
      $\lVert b' - 2^{i-\ell} \cdot b \rVert_\infty \le 4H$.
      Recall, the dynamic programs computes values
      for each element in $D_0, D_{1},\dotsc, D_{\ell}$.
      More precisely, for the value of $b'\in D_i$ we consider
      vectors $b''$ such that $b'', b'-b'' \in D_{i-1}$ and
      take the maximum sum of the values for $b'', b'-b''$ among all.
      For illustration consider the case of $m=1$. 
      Here we have that $b'\in D_i$ is equivalent to $-4H \le b'- 2^{i-\ell} \cdot b  \le 4H$. It is not hard to see that then
      the problem can be formulated as the following well-studied problem.
      \begin{definition}[{(min,~+)-convolution}]
      Given input variables $r_1,\dotsc,r_n \in \mathbb R$ and $s_1,\dotsc,s_n \in \mathbb R$, compute
      $t_1,\dotsc,t_n \in \mathbb R$, where
       $t_k = \min_{i + j = k} r_i + s_j$.
      \end{definition}
     We can also define (max,~+)-convolution as the counterpart where the maximum is
     taken instead of the minimum. The two problems are equivalent as
     each of them can be transformed to the other
     by negating the elements.
     There is a trivial $O(n^2)$ time algorithm for (min,~+)-convolution and it has been conjectured that
     there exists no truly sub-quadratic algorithm~\cite{DBLP:conf/icalp/CyganMWW17}. There does, however, exist an $O(n^2 / \log(n))$
     time algorithm~\cite{DBLP:journals/algorithmica/BremnerCDEHILPT14}, which we are going to use. In fact, there is an even faster algorithm
     that runs in time $O(n^2 / 2^{\Omega(\sqrt{\log (n)})})$~\cite{DBLP:conf/stoc/ChanL15}.

     Also when $m > 1$ the task of deriving $D_i$ from $D_{i-1}$ can
     be reformulated as a (min,~+)-convolution instance.
     For this, the $m$ dimensions of each $b'\in D_i$ are embedded
     in a single dimension with appropriate zero padding between them.
     The precise construction and its proof of correctness
     require some tedious calculations, which are deferred to the appendix.
     Using an algorithm for (min,~+)-convolution with running time $T(n)$
     we get an algorithm for ILP with running time
     $T(O(H)^m) \cdot \log(K)$. Inserting $T(n) = n^2/\log(n)$ and
     using $H\ge \Delta / 2$, we slightly improve on (\ref{eq:rt-naive})
     and obtain a running time of
     \begin{equation}\label{eq:rt-conv}
         O(H)^{2m} \cdot \frac{\log(K)}{\log(\Delta)}.
     \end{equation}
     Even more interesting though, a sub-quadratic algorithm for
     (min,~+)-convolution, where $T(n) = n^{2 - \delta}$ for some
     $\delta > 0$, would directly improve the exponent. Next, we
     will consider the problem of only testing feasibility of an ILP.
     Since we only record whether or not there exists a solution for
     a particular right-hand side, the convolution problem reduces to the following.
     \begin{definition}[\textsc{Boolean convolution}]
     Given input variables $r_1,\dotsc,r_n\in \{0, 1\}$ and $s_1,\dotsc,s_n\in\{0, 1\}$ compute
     $t_1,\dotsc,t_n\in \{0, 1\}$,
     where $t_k = \bigvee_{i + j = k} r_i \land s_j$.
     \end{definition}
     
     This problem can be solved very efficiently via fast Fourier transform. We compute the $(+,\cdot)$-convolution of the input.
     It is well known that this can be done using FFT
     in time $O(n\log(n))$.
     The $(+,\cdot)$-convolution of $r$ and $s$ is the vector $t$,
     where $t_k = \sum_{i+j = k} r_i \cdot s_j$.
     To get the Boolean convolution instead, we
     simply replace each $t_k > 0$ by $1$.
     Using $T(n) = O(n\log(n))$ for the convolution algorithm yields
     that a feasible solution can be found in time
     \begin{equation}\label{eq:rt-feas}
         O(H)^{m} \cdot \log(\Delta) \cdot \log(K) .
     \end{equation}
    \subsection{Proximity}\label{sec-proximity}
    Eisenbrand and Weismantel gave the following bound on the proximity 
    between fractional and integral solutions.
    \begin{theorem}[\cite{DBLP:journals/talg/EisenbrandW20}]
      \label{th-proximity}
      Let $\max\{c^T x : Ax = b, x\in\mathbb Z^n_{\ge 0}\}$ be feasible and bounded.
      Let $x^*$ be an optimal vertex solution of the fractional relaxation. Then there exists
      an optimal solution $z^*$ with
      \begin{equation*}
        \lVert z^* - x^* \rVert_1 \le m (2m\Delta+1)^m.
      \end{equation*}
    \end{theorem}
    We use the theorem to bound the value of $K$ at the expense of
    computing the optimum of the fractional relaxation. This follows a similar approach as used in~\cite{DBLP:journals/talg/EisenbrandW20}.
    Note that $z^*_i \ge \ell_i := \max\{0, \lceil x^*_i \rceil - m(2m\Delta + 1)^m\}$.
    By setting $y = x - \ell$ we obtain
    the equivalent ILP $\max\{c^T y : Ay = b - A\ell, y\in\mathbb Z_{\ge 0}^n\}$.
    It suffices to find an optimal solution to it. Notice that $z^* - \ell$
    is optimal for this ILP and we can bound
    \begin{equation*}
    \lVert z^* - \ell \rVert_1 \le \lVert z^* - x^* \rVert_1 + \lVert \ell - x^* \rVert_1 \le m(2m\Delta + 1)^m + m^2 (2m\Delta + 1)^m = O(m\Delta)^m .
    \end{equation*}
    Here, we use that $x^*$ and $\ell$ can only differ in the $m$ many non-zero components of $x^*$ and in those by at most $m(2m\Delta + 1)^m$.
    Also, note that the O-notation hides polynomial terms in $m$.
    Using $K = O(m\Delta)^m$, $H\le O(\sqrt{m} \Delta)$ and the 
    $O(m)^{m/2} \log^{3m}(m) \cdot n$ time
    algorithm~\cite{DBLP:journals/talg/CHAN18} for solving the
    relaxation, we derive a running time of
    \begin{multline}
        O(m)^{m/2} \log^{3m}(m) \cdot n + O(H)^{2m} \cdot \frac{\log(K)}{\log(\Delta)} \\
        \le O(\sqrt{m})^{m} \log^{3m}(m) \cdot O(\Delta)^m + O(\sqrt{m} \Delta)^{2m} \cdot \frac{m\log(m\Delta)}{\log(\Delta)} \le O(\sqrt{m}\Delta)^{2m} .
    \end{multline}
    Similarly, we can improve the running time for finding any feasible
    solution to
    \begin{multline}
        O(m)^{m/2} \log^{3m}(m) \cdot n + O(H)^{m} \cdot \log(\Delta) \log(K) \\
        \le O(\sqrt{m})^{m} \log^{3m}(m) \cdot O(\Delta)^m + O(\sqrt{m} \Delta)^{m} \log(\Delta)\cdot m\log(m\Delta)
        \le O(\sqrt{m}\Delta)^{(1 + o(1))m} .\label{eq:rt-feas-prox}
    \end{multline}
    This proves the running times (\ref{eq:opt}) and (\ref{eq:feas})
    in the case that the ILP is bounded. Testing whether
    an ILP is unbounded can be done without increasing the asymptotic
    running time, as we will lay out next.
    
    \subsection{Unbounded solutions}\label{sec-unbounded}
      The ILP $\max\{c^Tx : Ax = b, x\in\mathbb Z_{\ge 0}^n\}$ is unbounded, if and only if it is feasible and
      $\max\{c^T x : Ax = 0, x\in\mathbb Z_{\ge 0}^n\}$ has a solution with positive value. The former can be checked with our algorithm,
      hence it remains to check if the latter condition holds.
      We can simply solve the LP relaxation for this. If there is
      a fractional solution with positive value, there is also an
      integral one. This is because by Cramer's rule there
      exists a fractional solution with denominators $\det(A)$,
      hence multiplying by $\det(A)$ yields an integral solution.
      
  \subsection{Heterogeneous rows}
    Let $\Delta_1,\dotsc,\Delta_m\le \Delta$ denote the largest absolute
    values of each row in $A$. When some of these values are much smaller
    than $\Delta$, the maximum among all,
    we can do better than $O(\sqrt m\Delta)^{2m}$.
    Define $A' = \mathrm{diag}(\Delta_1^{-1},\dotsc,\Delta_m^{-1}) \cdot A$, where
\begin{equation*}
  \mathrm{diag}(\Delta_1^{-1},\dotsc,\Delta_m^{-1}) =
      \begin{pmatrix}
        \Delta_1^{-1} & & 0 \\
            & \ddots & \\
        0   & & \Delta_m^{-1}
      \end{pmatrix} .
\end{equation*}
    We claim that in the dynamic program a table of size $\prod_{k=1}^m O(H' \Delta_k)$ suffices, where $H' \ge \mathrm{herdisc}(A')$.
Clearly, the ILP $\max\{c^T x, Ax = b, x\in\mathbb Z_{\ge 0}^n\}$ is equivalent to
\begin{equation*}
  \max\{c^T x, A' x = b', x\in\mathbb Z_{\ge 0}^n\} ,
\end{equation*}
    where $b' = \mathrm{diag}(\Delta_1^{-1},\dotsc,\Delta_m^{-1}) \cdot b$.
    At first glance, our algorithm cannot be applied to this problem, since the entries are
    not integral.
    However, in the algorithm we only use the fact that the number of points
    $Ax$ with $x\in\mathbb Z_{\ge 0}^n$ close to some point $b''$, that is,
    with $\lVert Ax - b'' \rVert_\infty \le 4H$, is small and can be enumerated.
    The points $A' x$ with $x\in\mathbb Z_{\ge 0}^n$ and $\lVert A' x - b'' \rVert_\infty \le 4H'$
    are exactly those with $| (A x)_k - b''_k \cdot \Delta_k | \le 4H' \cdot \Delta_k$ for all $k$.
    These are $\prod_{k=1}^m O(H' \Delta_k)$ many and they can be enumerated.
    This way, we get a running time of $\prod_{k=1}^m O(H'\Delta_k)^2$,
    which using the bound from Theorem~\ref{th:spencer} yields
  \begin{equation}
    \prod_{k=1}^m O(m \Delta_k^2) .
  \end{equation}
    If one is only interested in a feasible solution, then this improves to
    \begin{equation}
        O(\sqrt{m})^{(1 + o(1))m} \cdot \prod_{k=1}^m [m \Delta_k] \cdot \log^2(\Delta) .
    \end{equation}
  \section{Lower bounds}
  In this section we give conditional lower bounds that match
  the running time of our algorithm both for finding an optimal solution
  and for finding a feasible solution.
  \subsection{Optimization problem}\label{sec-lb-opt}
  We use an equivalence between the problems Unbounded Knapsack and (min,~+)-convolution
  regarding sub-quadratic algorithms.
  \begin{definition}[Unbounded Knapsack]\label{def:unbounded-knapsack}
  Given $C\in\mathbb N$, $w_1,\dotsc,w_n\in \mathbb N$, and $p_1,\dotsc,p_n\in\mathbb N$
      find integer multiplicities $x_1,\dotsc,x_n$,
    such that $\sum_{i=1}^n x_i \cdot w_i \le C$ and $\sum_{i=1}^n x_i \cdot p_i$ is maximized.
  \end{definition}
  
  Note that when we instead require $\sum_{i=1}^n x_i \cdot w_i = C$ in the problem above, 
  we can transform it to this form by adding an item of profit zero and weight $1$.
  \begin{theorem}[\cite{DBLP:conf/icalp/CyganMWW17, DBLP:conf/icalp/KunnemannPS17}]
    For any $\delta > 0$ there exists no $O((n + C)^{2-\delta})$ time algorithm for Unbounded Knapsack unless there is a truly sub-quadratic algorithm for (min,~+)-convolution.
  \end{theorem}
  When using this theorem, we assume that the input already consists of
  the at most $C$ relevant items only, $n\le C$, and $w_i\le C$ for all $i$.
  This preprocessing can be done in time $O(n + C)$.
  \begin{theorem}\label{th-lower-bound}
    Let $m \in \mathbb N$. For any $\delta > 0$ and any computable function $f$
    there does not exist an algorithm that solves
    ILPs with $m$ constraints in time
    $f(m) \cdot (n^{2-\delta} + (\Delta+ \lVert b \rVert_\infty)^{2m - \delta})$,
    unless there exists a truly sub-quadratic algorithm for (min,~+)-convolution.
  \end{theorem}
  \begin{proof}{Proof.}
    Let $\delta > 0$ and $m\in\mathbb N$.
    Assume that there exists an algorithm that solves ILPs of the form $\max \{c^T x : Ax = b, x\in\mathbb Z_{\ge 0}^n\}$ where $A\in \mathbb Z^{m\times n}$, $b\in\mathbb Z^m$, and $c\in\mathbb Z^n$ in time
    $f(m) \cdot (n^{2-\delta} + (\Delta+ \lVert b \rVert_\infty)^{2m - \delta})$,
    where $\Delta$ is the greatest absolute value in $A$.
    We will show that this implies an $O((n+C)^{2 - \delta'})$ time algorithm for 
    Unbounded Knapsack for some $\delta' > 0$.
    Let $(C, (w_i)_{i=1}^n, (p_i)_{i=1}^n)$ be an instance of this problem.
    Let us first observe that the claim holds for $m=1$. Clearly
    Unbounded Knapsack (with equality) can be written
    as the following ILP.
    \begin{align*}
      \max \sum_{i=1}^n & p_i\cdot x_i  \\
      \sum_{i=1}^n w_i \cdot x_i &= C \tag{UKS1} \\
      x &\in \mathbb Z_{\ge 0}^n
    \end{align*}
    Since $w_i\le C$ for all $i$ (otherwise the item can be discarded), we can solve this ILP
    by assumption in time
    $f(1) \cdot (n^{2-\delta} + (2 C)^{2-\delta}) \le O((n + C)^{2-\delta})$.
    Now consider the case where $m > 1$. We want to reduce $\Delta$ by exploiting the additional
    rows. Let $\Delta = \lfloor C^{1/m} \rfloor + 1 > C^{1/m}$. We write $C$ 
    in base-$\Delta$ notation, that is,
    \begin{equation*}
      C = C^{(0)} + \Delta C^{(1)} + \cdots + \Delta^{m-1} C^{(m-1)} ,
    \end{equation*}
    where $0\le C^{(k)} < \Delta$ for all $k$.
    Likewise, write
    $w_i = w_i^{(0)} + \Delta w_i^{(1)} + \cdots + \Delta^{m-1} w_i^{(m-1)}$
    with $0 \le w_i^{(k)} < \Delta$ for all $k$.
    We claim that (UKS1) is equivalent to the following ILP.
    \begin{align}
      \max \sum_{i=1}^n p_i\cdot x_i & \notag\\
      \sum_{i=1}^n [w^{(0)}_i \cdot x_i] - \Delta \cdot y_1 &= C^{(0)} \label{UKSm-1}\\
      \sum_{i=1}^n [w^{(1)}_i \cdot x_i] + y_1 - \Delta \cdot y_2 &= C^{(1)} \label{UKSm-2}\\
       &\vdots \tag{UKSm} \\
      \sum_{i=1}^n [w^{(m-2)}_i \cdot x_i] + y_{m-2} - \Delta \cdot y_{m-1} &= C^{(m-2)} \label{UKSm-3}\\
      \sum_{i=1}^n [w^{(m-1)}_i \cdot x_i] + y_{m-1} &= C^{(m-1)} \label{UKSm-4}\\
      x \in \mathbb Z_{\ge 0}^n & \notag \\
      y \in \mathbb Z_{\ge 0}^{m-1} & \notag
    \end{align}
    \paragraph*{Implication $x\in (\mathrm{USK1}) \Rightarrow x\in (\mathrm{USKm})$.}
    Let $x$ be a solution to (UKS1).
    Then for all $1\le\ell\le m$,
    \begin{equation*}
      \sum_{i=1}^n \sum_{k=0}^{\ell-1} \Delta^k w_i^{(k)} \cdot x_i \equiv \sum_{i=1}^n w_i \cdot x_i \equiv C \equiv \sum_{k=0}^{\ell-1}\Delta^k C^{(k)}\mod \Delta^{\ell} .
    \end{equation*}
    This is because all $\Delta^\ell w_i^{(\ell)},\dotsc, \Delta^{m-1} w_i^{(m-1)}$ and $\Delta^\ell C^{(\ell)},\dotsc,\Delta^{m-1} C^{(m-1)}$
    are multiples of $\Delta^\ell$.
    It follows that there exists an $y_\ell\in \mathbb Z$ such that 
    \begin{equation*}
  \sum_{i=1}^n \sum_{k=0}^{\ell-1} [\Delta^k w_i^{(k)} \cdot x_i] - \Delta^\ell\cdot y_\ell
  = \sum_{k=0}^{\ell-1} \Delta^k C^{(k)} .
    \end{equation*}
    Furthermore, $y_\ell$ is non-negative, because otherwise
    \begin{multline*}
    \sum_{k=0}^{\ell-1} \Delta^k C^{(k)}
    \le \sum_{k=0}^{\ell-1} \Delta^k (\Delta - 1)
    < \Delta^{\ell-1} (\Delta - 1)\sum_{k=0}^{\infty} \Delta^{- k} \\
    = \Delta^{\ell-1} \frac{\Delta - 1}{1 - \frac 1 \Delta} = \Delta^\ell
    \le - \Delta^\ell y_\ell
    \le \sum_{i=1}^n \sum_{k=0}^{\ell-1} [\Delta^k w_i^{(k)} \cdot x_i]
    - \Delta^\ell y_\ell .
    \end{multline*}
    We choose $y_1,\dotsc,y_{m}$ exactly like this.
    The first constraint (\ref{UKSm-1}) follows directly. Now let 
    $\ell\in\{2,\dotsc, m\}$. By choice of $y_{\ell-1}$ and $y_\ell$ we
    have that
    \begin{equation}
      \sum_{i=1}^n \bigg[ \underbrace{\left(\sum_{k=0}^{\ell-1} \Delta^k w_i^{(k)} - \sum_{k=0}^{\ell-2} \Delta^k w_i^{(k)}\right)}_{= \Delta^{\ell-1} w^{(\ell-1)}_i} \cdot x_i \bigg]
      + \Delta^{\ell-1} \cdot y_{\ell - 1} - \Delta^\ell \cdot y_\ell
      = \underbrace{\sum_{k=0}^{\ell-1} \Delta^k C^{(k)} - \sum_{k=0}^{\ell-2} \Delta^k C^{(k)}}_{=\Delta^{\ell-1} C^{(\ell-1)}} .\label{constr-ell2}
    \end{equation}
    Dividing both sides by $\Delta^{\ell-1}$ we get every constraint (\ref{UKSm-2}) - (\ref{UKSm-3}) for the correct choice of $\ell$.
    Finally, consider the special case
    of the last constraint (\ref{UKSm-4}). By choice of $y_m$ we have that
    \begin{equation*}
  \sum_{i=1}^n \underbrace{\sum_{k=0}^{m-1} \Delta^k w_i^{(k)}}_{=w_i} \cdot x_i - \Delta^m\cdot y_m = \underbrace{\sum_{k=0}^{m-1} \Delta^k C^{(k)}}_{=C} .
    \end{equation*}
    Thus, $y_m = 0$ and (\ref{constr-ell2}) implies the last constraint (with $\ell = m$).

    \paragraph*{Implication $x\in (\mathrm{USKm}) \Rightarrow x\in (\mathrm{USK1})$.}
    Let $x_1,\dotsc,x_n, y_1,\dotsc,y_{m-1}$ be a solution to (UKSm) and set $y_m = 0$.
    We show by induction that for all $\ell\in\{1,\dotsc,m\}$ it holds that
    \begin{equation*}
      \sum_{i=1}^n \sum_{k=0}^{\ell-1}\Delta^k w_i^{(k)} \cdot x_i - \Delta^\ell y_\ell = \sum_{k=0}^{\ell-1}\Delta^k C^{(k)} .
    \end{equation*}
    With $\ell=m$ this implies the claim as $y_m = 0$ by definition.
    For $\ell = 1$ the equation is exactly the first constraint (\ref{UKSm-1}). Now let $\ell > 1$ and
    assume that the equation above holds. We will show that it also holds for $\ell+1$. From
    (USKm) we have
    \begin{equation*}
      \sum_{i=1}^n [ w_i^{(\ell)} \cdot x_i ]
       + y_{\ell} - \Delta \cdot y_{\ell + 1} = C^{(\ell)} .
    \end{equation*}
    Multiplying each side by $\Delta^{\ell}$ we get
    \begin{equation*}
      \sum_{i=1}^n [ \Delta^{\ell} w_i^{(\ell)} \cdot x_i ]
       + \Delta^{\ell} y_{\ell} - \Delta^{\ell+1} \cdot y_{\ell + 1} = \Delta^{\ell} C^{(\ell)} .
    \end{equation*}
    By adding and subtracting the same elements, it follows that
    \begin{equation*}
      \sum_{i=1}^n \left[ \bigg( \sum_{k=0}^{\ell} \Delta^k w_i^{(k)} -  \sum_{k=0}^{\ell-1} \Delta^k w_i^{(k)} \bigg) \cdot x_i \right]
       + \Delta^{\ell} \cdot y_{\ell} - \Delta^{\ell+1} \cdot y_{\ell + 1}
  = \sum_{k=0}^{\ell}\Delta^{k} C^{(k)} - \sum_{k=0}^{\ell-1}\Delta^{k} C^{(k)} .
    \end{equation*}
    By inserting the induction hypothesis we conclude
    \begin{equation*}
      \sum_{i=1}^n \sum_{k=0}^{\ell} [ \Delta^k w_i^{(k)} \cdot x_i ]
       - \Delta^{\ell+1} y_{\ell + 1} = \sum_{k=0}^{\ell}\Delta^{k} C^{(k)} .
    \end{equation*}
    \paragraph*{Constructing and solving the ILP.}
    The ILP (UKSm) can be constructed easily in $O(Cm + nm) \le O((n + C)^{2-\delta/m})$ operations (recall that $m$ is a constant).
    We obtain $\Delta = \lfloor C^{1/m} \rfloor + 1$ by guessing:
    More precisely, we iterate over all numbers $\Delta_0\le C$ and
    find the one where $(\Delta_0-1)^m < C \le \Delta_0^m$.
    Although there are
    more efficient, non-trivial ways to compute the rounded $m$-th root,
    this is not required here.
    The base-$\Delta$ representation for $w_1,\dotsc,w_n$ and $C$ can be computed with $O(m)$ operations
    for each of these numbers.

    All entries of the matrix in (UKSm) and the right-hand side are bounded by $\Delta = O(C^{1/m})$. Therefore, by assumption this ILP can be solved in time
    \begin{equation*}
  f(m) \cdot (n^{2-\delta} + O(C^{1/m})^{2m - \delta})
   \le f(m) \cdot O(1)^{2m-\delta} \cdot (n + C)^{2 - \delta / m} = O( (n + C)^{2 - \delta/m}) .
    \end{equation*}
    This yields a truly sub-quadratic algorithm for Unbounded Knapsack.   \end{proof}
  \subsection{Feasibility problem}
   We will show that our algorithm for solving feasibility of ILPs is optimal (except for sub-polynomial improvements).
   We use a recently discovered lower bound for k-SUM based on the SETH.
   \begin{definition}[k-SUM]
      Given $T\in \mathbb N_0$ and
        $Z_1,\dotsc, Z_k \subset \mathbb N_0$ where $|Z_1| + |Z_2| + \cdots + |Z_k| = n\in\mathbb N$
      find $z_1\in Z_1, z_2\in Z_2, \dotsc, z_k\in Z_k$ such that
        $z_1 + z_2 + \cdots + z_k = T$.
   \end{definition}

  \begin{theorem}[\cite{DBLP:conf/soda/AbboudBHS19}]\label{SETH-kSUM}
    If the SETH holds,
    then for every $\delta > 0$ there exists a value $\gamma > 0$ such that k-SUM cannot be solved in time
    $O(T^{1 - \delta} \cdot n^{\gamma k})$.
  \end{theorem}
  This implies that for every $p\in\mathbb N$ there is no $O(T^{1-\delta} \cdot n^p)$ time algorithm
  for k-SUM if $k \ge p/\gamma$.
  \begin{theorem}\label{th-lower-bound2}
    Let $m\in \mathbb N$.
    If the SETH holds, then for every $\delta > 0$ and every computable function $f$,
    there does not exist an algorithm that solves
    feasibility of ILPs with $m$ constraints in time $n^{f(m)} \cdot (\Delta + \lVert b \rVert_\infty)^{m - \delta}$.
  \end{theorem}
  \begin{proof}{Proof.}
    Like in the previous reduction we start with the case of $m=1$. For higher values of $m$ the
    result can be shown in the same way as before.

    Suppose there exists an algorithm for solving feasibility of ILPs with one constraint
    in time $n^{f(1)} \cdot (\Delta + \lVert b \rVert_\infty)^{1-\delta}$ for some $\delta > 0$ and $f(1)\in\mathbb N$.
    Let $\gamma$ be the constant given by Theorem~\ref{SETH-kSUM} for this $\delta$ and
    set $k = \lceil f(1)/\gamma\rceil$.
    Now consider an 
    instance $(T, Z_1,\dotsc, Z_k)$ of k-SUM. We will show that this can be solved in
    $O(T^{1-\delta}\cdot n^{f(1)})$, which contradicts the SETH.
    For every $i\le k$ and every $z\in Z_i$ we use a binary variable $x_{i, z}$ that describes whether
    $z$ is used. We can easily model k-SUM as the following ILP:
    \begin{align*}
      \sum_{i=1}^k\sum_{z\in Z_i} z \cdot x_{i, z} &= T & \\
      \sum_{z\in Z_i} x_{i, z} &= 1 &\forall i\in\{1,\dotsc, k\} \\
      x_{i, z} &\in \mathbb Z_{\ge 0} &\forall i\in\{1,\dotsc,k\}, z\in Z_i
    \end{align*}
    However, since we want to reduce to an ILP with one constraint, 
    we need a slightly more sophisticated construction.
    We will show that the cardinality constraints can be encoded into the k-SUM instance 
    by increasing the numbers by a factor of $2^{O(k)}$, which is in $O(1)$
    since $k$ is some constant depending on $f(1)$ and $\gamma$ only.
    We will use this to obtain an ILP with only one
    constraint and values of size at most $O(T)$.
    A similar construction is also used in~\cite{DBLP:conf/soda/AbboudBHS19}.

    Our goal is to construct an instance $(T', Z'_k,\dotsc, Z'_k)$ such that
    for every $x^*$ it holds that 
    $x^*$ is a solution to the first ILP if and only if
    \begin{equation}\label{eq:feas-cond}
        x^*\in \{x : \sum_{i=1}^k\sum_{z\in Z'_i} z \cdot x_{i, z} = T', x\in \mathbb Z_{\ge 0}^n\} .
    \end{equation}
    We will use one element to represent each element in the original instance.
    Consider the binary representation of numbers in $Z'_1\cup\cdots\cup Z'_k$ and of $T'$.
    The numbers in the new instance will consist of three parts and $\lceil\log(k)\rceil$
    many 0s between them to prevent interference.
    For an illustration of the construction see Figure~\ref{fig-construction}.
    The $\lceil\log(k)\rceil$ most significant bits ensure that exactly $k$ elements are selected;
    the middle part are $k$ bits that ensure of every set $Z'_i$ exactly one element is selected;
    the least significant $\lceil \log (T) \rceil$ bits represent the original values of the elements.
    Set the values in the first part of the numbers to $1$ for all elements
    $Z'_1\cup\cdots\cup Z'_k$ and to
    $k$ in $T'$. Clearly this ensures that at most $k$ elements are chosen.
    The sum of at most $k$ elements cannot be larger than $k \le 2^{\lceil\log(k)\rceil}$
    times the biggest element. This implies that the buffers of $\lceil\log(k)\rceil$ zeroes cannot
    overflow and we can consider each of the three parts independently.
    It follows that exactly $k$ elements must be chosen by any feasible solution.
    The system $\{x : \sum_{i=1}^k 2^i x_i = 2^{k+1} - 1, \lVert x \rVert_1 = k, \mathbb Z_{\ge 0}^k\}$
    has exactly one solution and this solution is $(1,1,\dotsc,1)$:
    Consider summing up $k$ powers of $2$ and envision the binary representation
    of the partial sums. When we add some $2^i$ to the partial sum,
    the number of ones in the binary representation increases by one, if the
    $i$'th bit of the current sum is zero. Otherwise, it does not increase.
    However, since in the binary representation of the final sum there are
    $k$ ones, it has to increase in each addition. This means no power of
    two can be added twice and therefore each has to be added exactly once.
    
    It follows that the second part of the numbers
    enforces that of every $Z'_i$ exactly one
    element is chosen.
    We conclude that (\ref{eq:feas-cond}) solves the initial k-SUM instance. By assumption this can be
    done in time $n^{f(1)} \cdot (\Delta + \lVert b \rVert_\infty)^{1 - \delta} = n^{f(1)} \cdot O(T')^{1 - \delta} = O(n^{f(1)} \cdot T^{1 - \delta})$. Here we use that $T' \le 2^{3\log(k) + k + \log(T) + 4} = O(k^3 2^k T) = O(T)$, since $k$ is a constant.

    For $m > 1$ we can use the same construction as in the reduction for the optimization problem:
    Suppose there is an algorithm that finds feasible solutions to ILPs with $m$ constraints in
    time $n^{f(m)} \cdot (\Delta + \lVert b \rVert_\infty)^{m - \delta}$.
    Choose $\gamma$ such that there is no algorithm for k-{SUM} with running time
    $O(T^{1-\delta/m} \cdot n^{\gamma k})$ (under SETH). We set $k = \lceil f(m) / \gamma \rceil$.
    By splitting the one constraint of (\ref{eq:feas-cond}) into $m$ constraints we can reduce the upper bound on
    elements from $O(T)$ to $O(T^{1/m})$.
    This means the assumed running time for solving ILPs can be used to solve k-{SUM} in time
    \begin{equation*}
      n^{f(m)} \cdot O(T^{1/m})^{m - \delta}
      \le n^{\gamma k} O(1)^{m-\delta} T^{1 - \delta / m} = O(n^{\gamma k} T^{1 - \delta / m}) . \qedhere
    \end{equation*}
\begin{figure}
  \begin{align*}
    Z'_i \ni z' &= \overbrace{\underbracket{0 \dots 0 0 0 1}_{\lceil\log(k)\rceil}}^{\bin(1)}
     | \underbracket{0 \dots 0}_{\lceil\log(k)\rceil}
     | \overbrace{\underbracket{0 \dots 0 1 0 \dots 0}_{k}}^{\bin(2^i)}
     | \underbracket{0 \dots 0}_{\lceil\log(k)\rceil} |
    \overbrace{\underbracket{0 1 1 0 \dots}_{\lceil \log(T) \rceil}}^{\bin(z)} \\
    T' &= \overbrace{\underbracket{0 \dots 1 0 1 1}_{\lceil\log(k)\rceil}}^{\bin(k)}
     | \underbracket{0 \dots 0}_{\lceil\log(k)\rceil}
     | \overbrace{\underbracket{1 1 1 1 \dots 1 1 1 1}_{k}}^{\bin(2^{k+1} - 1)}
     | \underbracket{0 \dots 0}_{\lceil\log(k)\rceil} |
    \overbrace{\underbracket{1 0 1 1 \dots}_{\lceil \log(T) \rceil}}^{\bin(T)}
  \end{align*}
  \caption{Construction of $Z'_i$ and $T'$}
  \label{fig-construction}
\end{figure}
  \end{proof}
  \section{Applications}
  In this section we apply our results to some well-known problems,
  which can be formulated using ILPs with few constraints and small entries.
  In particular, we give examples, where the reduction of the running time
  by a factor $n$ improves on the state-of-the-art, the heterogeneity of rows plays a role, and one
  where the use of hereditary discrepancy and the removal of the
  dependence on $\lVert b \rVert_\infty$ are relevant.
  \subsection{Unbounded Knapsack and Unbounded Subset-Sum}
  Recall Definition~\ref{def:unbounded-knapsack}, which introduces the problem Unbounded Knapsack.
  Traditionally, $C$ is only an upper bound on $\sum_{i=1}^n w_i \cdot x_i$ in most of the literature,
  but that variant easily reduces to the problem above by adding a slack variable.
  Unbounded Subset-Sum is the same problem without an objective function, i.e.,
  the problem of finding a multi-set of items whose weights $w_i$ sum up to exactly $C$.
  We assume that no two items have the same weight. Otherwise in time $O(n + \Delta)$ we
  can remove all duplicates by keeping only the most valuable ones.
  This gives algorithms with running time $O(n + \Delta^2)$
  and $O(n + \Delta \log^2(\Delta))$ for Unbounded Knapsack and Unbounded Subset-Sum, respectively, where $\Delta$ is the
  maximum weight among all items (using the results from Section~\ref{sec-proximity}).
  The previously best pseudo-polynomial algorithms for Unbounded Knapsack,
  have running times $O(nC)$ (standard dynamic programming; see e.g.~\cite{DBLP:books/daglib/0010031}),
  $O(n\Delta^2)$~\cite{DBLP:journals/talg/EisenbrandW20},
  or very recently $O(\Delta^2\log(C))$~\cite{axiotis2018capacitated}.
  We note that the last algorithm, which was 
  discovered simultaneously and independently to ours,
  follows a very similar approach to ours when restricted to the Unbounded
  Knapsack case. After our work Chan and He gave an interesting
  improvement, which achieves a running time of $O(n\Delta \log^3(\Delta))$~\cite{DBLP:journals/jcss/ChanH22}. Note that
  $n$ is potentially much smaller than $\Delta$, but not vice versa
  
  For Unbounded Subset-Sum the state-of-the-art
  is a $O(C\log(C))$ time algorithm~\cite{DBLP:conf/soda/Bringmann17}.
  Hence, our algorithm is preferable when
  $\Delta \ll C$. Very recently Klein~\cite{klein2022fine} studied this problem
  and showed the perhaps surprising fact that there is also a pseudo-polynomial algorithm in
  terms of the smallest weight (and not the largest),
  but then the dependence on it is quadratic and cannot be improved
  unless the (min,~+)-convolution conjecture is false.
  \subsection{Change Making}
  In the Change Making problem we are given an infinite supply of coins with
  values $c_1 < c_2 < \cdots < c_n$ and a target t. The goal is to match $t$
  with as few coins as possible. In the decision variant, where we want to find a solution
  with at most $k$ coins, this can be written as finding a solution to the ILP
  \begin{equation*}
      \left\{\sum_{i=1} c_i x_i = t, \sum_{i=1}^n x_i + s = k, x\in \mathbb Z_{\ge 0}, s\in\mathbb Z_{\ge 0}\right\} .
  \end{equation*}
  In other words, this is a feasibility ILP with two rows, where the first row has maximum
  coefficient $c_n$ and the second row has maximum coefficient $1$. Using~(\ref{eq:rt-feas-prox})
  this can be solved in time
  \begin{equation*}
      O(c_n \log^2(c_n)).
  \end{equation*}
  This matches exactly the running time in~\cite{DBLP:journals/jcss/ChanH22}. In fact,
  that algorithm behaves very similar to ours when restricted to this problem.
  \subsection{Scheduling Jobs on Identical Parallel Machines}
  The problem Scheduling Jobs on Identical Parallel Machines asks
  to distribute
  $N$ jobs onto $M \le N$ machines. Each job $j$ has a processing time $p_j$ and the
  objective is to minimize the makespan, that is, the maximum sum of processing times on a single
  machine. Since an exact solution cannot be computed unless $\mathrm P = \mathrm{NP}$, we
  are satisfied with a $(1+\epsilon)$-approximation, where $\epsilon > 0$ is part of the input.
  We will outline how this problem can be solved using our algorithm.
  This gives the best known running time, which is even a slight improvement over
  the sophisticated algorithm for this problem in~\cite{DBLP:journals/mor/JansenKV20}.

  We consider here the variant, in which a makespan $\tau$ is given and we have to find
  a schedule with makespan at most $(1 + \epsilon)\tau$ or prove that there exists no schedule
  with makespan at most $\tau$. This suffices by using a standard dual approximation framework.
It is easy to see that one can discard all jobs of size at most $\epsilon \cdot \tau$ and add them
greedily after a solution for the other jobs is found. The big jobs can each be rounded
to the next value of the form $\epsilon\cdot \tau \cdot (1 + \epsilon)^i$ for some $i$.
This reduces the number of different processing times to $O(1/\epsilon\log(1/\epsilon))$ many and
increases the makespan by at most a factor of $1 + \epsilon$.
We are now ready to write this problem as an ILP.
A configuration is a way to use a machine. It describes how many jobs of each size are assigned
to this machine. Since we aim for a makespan of $(1 + \epsilon)\cdot \tau$, the sum of these
sizes must not exceed this value. The configuration ILP has a variable for every valid configuration
and it describes how many machines use this configuration.
Let $\mathcal C$ be the set of valid configurations and $C_k$ the multiplicity of size $k$
in a configuration $C\in\mathcal C$.
The following ILP solves the rounded instance. We note that there is no objective function in it.
\begin{align*}
  \sum_{C\in\mathcal C} x_{C} &= M \\
  \sum_{C\in\mathcal C} C_k \cdot x_{C} &= N_k & \forall k\in\mathcal K \\
  x_{C} &\in \mathbb Z_{\ge 0} & \forall C\in\mathcal C
\end{align*}
Here $\mathcal K$ are the rounded sizes and $N_k$ the number of jobs with rounded size $k\in\mathcal K$.
The first constraint enforces that the correct number of machines is used, the next $|\mathcal K|$ many
enforce that for each size the correct number of jobs is scheduled.

It is notable that this ILP has only few constraints (a constant for a fixed choice of $\epsilon$) and also
the $\ell_1$-norm of each column is small. More precisely, it is at most $1/\epsilon$, since every size is at least $\epsilon \cdot \tau$ and therefore no more than $1/\epsilon$ jobs fit in one configuration.
  By the Theorem~\ref{beck-fiala} we know that $H = 1/\epsilon$ is an upper bound on the hereditary discrepancy, $\Delta \le 1/\epsilon$,
  $m = O(1/\epsilon\log(1/\epsilon))$, $\lVert b \rVert_\infty \le N$, and $n\le (1/\epsilon)^{O(1/\epsilon\log(1/\epsilon))}$. Notice also that $K = N$ is
  a trivial upper bound on the $\ell_1$-norm of any solution.
  Using~(\ref{eq:rt-feas}) and rounding in time $O(N + 1/\epsilon \log(1/\epsilon))$
  yields a running time of
  \begin{multline*}
    O(H)^m \log(\Delta)\log(K) + O(nm) + 
    O\left(N + \frac{1}{\epsilon}\log\left(\frac{1}{\epsilon}\right)\right) \\
    \le 2^{O(1/\epsilon\log^2(1/\epsilon))}\log(N) + O\left(N + \frac{1}{\epsilon}\log\left(\frac{1}{\epsilon}\right)\right)
    \le 2^{O(1/\epsilon\log^2(1/\epsilon))} + O(N) .
  \end{multline*}
  The inequality above follows from distinguishing between
  $2^{O(1/\epsilon\log^2(1/\epsilon))} \le \log(N)$ and
  $2^{O(1/\epsilon\log^2(1/\epsilon))} > \log(N)$.
  The same running time (except for a higher constant in the exponent)
  could be obtained with~\cite{DBLP:journals/talg/EisenbrandW20}.
  However, in order to avoid a multiplicative factor of $N$,
  one would have to solve the LP relaxation first and then use proximity.
  Our approach gives an easier, purely combinatorial algorithm.
  The advantage of our algorithm comes from removing the dependence on $\lVert b \rVert_\infty$. Recently, the authors together with Berndt and Deppert~\cite{berndt2022load}
  introduced a more involved ILP for this problem, which reduces the $\ell_1$-norm
  of each column to $O(\log(1/\epsilon))$ while maintaining the other bounds.
  Still using the algorithm from this work, this leads to a mild improvement of
  the running time to
  \begin{equation*}
      2^{O(1/\epsilon \log(1/\epsilon)\log\log(1/\epsilon))} + O(N) .  
  \end{equation*}
  This improvement relies on the low hereditary discrepancy and does not follow with
  the weaker bounds on the Steinitz Lemma as in~\cite{DBLP:journals/talg/EisenbrandW20}.


\bibliographystyle{plain} 
\bibliography{main.bib} 


\appendix

\section*{Omitted proofs}

\paragraph*{Splitting a solution into two even parts.}
Recall that Lemma~\ref{lem-split} says that a vector $x$ can be split into two
parts $z$ and $x - z$ such that $Az, A(x - z) \approx 1/2 \cdot Ax$ and the
$\ell_1$-norm of each part is at least a constant fraction of that of $x$.
\begin{proof}{Proof of Lemma~\ref{lem-split}}
Let $x'_i = \lfloor x_i / 2 \rfloor$ and
$x''_i = \lceil x_i / 2 \rceil - \lfloor x_i / 2 \rfloor \in \{0, 1\}$ for all $i$.
Then $x_i = \lfloor x_i/2 \rfloor + \lceil x_i/2 \rceil = 2x_i' + x_i''$.
Now apply the definition of $\mathrm{disc}(A_I)$ to $x''$, where $I=\mathrm{supp}(x'')$ are the
indices $i$ with $x''_i = 1$.
This way we obtain a vector $z''\in\{0, 1\}^n$ with
$\lVert A (z'' - x''/2) \rVert_\infty \le \mathrm{disc}(A_I) \le \mathrm{herdisc}(A)$.
We now use $z = x' + z''$ to show the first part of the lemma.
Then
\begin{equation*}
  \left\lVert A \left(z - \frac x 2\right) \right\rVert_\infty
  = \left\lVert A \left((x' + z'') - \frac {2x' + x''} 2\right) \right\rVert_\infty 
  = \left\lVert A \left(z'' - \frac {x''} 2\right) \right\rVert_\infty \le \mathrm{herdisc}(A).
\end{equation*}
Furthermore, for all $i$
\begin{equation*}
  0 \le \underbrace{x'_i + z''_i}_{=z_i} \le x'_i + x''_i \le 2x'_i + x''_i = x_i .
\end{equation*}
In order to control the $\ell_1$ norm in the second part of the lemma,
we first split $x$ into two non-empty
$y', y''\in\mathbb Z_{\ge 0}^n$ with $y' + y'' = x$ and
$\lfloor \lVert x \rVert_1 / 2\rfloor = \lVert y' \rVert_1 \le \lVert y'' \rVert_1 = \lceil \lVert x \rVert_1 / 2 \rceil$.
Now apply the first part of the lemma to obtain $z' \le y'$ and $z'' \le y''$
with $\lVert A(z' - y'/2) \rVert_\infty \le \mathrm{herdisc}(A)$ and
$\lVert A(z'' - y''/2) \rVert_\infty \le \mathrm{herdisc}(A)$.
We can assume w.l.o.g. that $\lVert z' \rVert_1 \le \lVert y' \rVert_1 / 2 \le \lVert x \rVert_1 / 4$ and
$\lVert z'' \rVert_1 \ge \lVert y'' \rVert_1 / 2 \ge \lVert x \rVert_1 / 4$, since otherwise we can
swap them for $y' - z'$ and $y'' - z''$, respectively.
We will use $z = z' + z''$ for the second part of the lemma.
As for the lower bound,
\begin{equation*}
  \lVert z \rVert_1
  \ge \lVert z'' \rVert_1 \ge \frac{\lVert x \rVert_1} 4 .
\end{equation*}
For the upper bound we first consider
the case where $\lVert x \rVert_1 \le 5$
and note that $\lVert z' \rVert_1 \le \lVert y' \rVert_1 / 2 = \lVert y' \rVert_1 - \lVert y' \rVert_1 / 2 < \lVert y' \rVert_1$.
Thus,
\begin{equation*}
  \lVert z \rVert_1 = \lVert z' + z'' \rVert_1 \le \lVert y' + y'' \rVert_1 - 1
  \le \lVert x \rVert_1 - \frac 1 5  \lVert x \rVert_1 = \frac 4 5 \lVert x \rVert_1 .
\end{equation*}
If $\lVert x \rVert_1 \ge 6$,
\begin{equation*}
  \lVert z \rVert_1 = \lVert z' + z'' \rVert_1
  \le \frac{\lVert x \rVert_1}{4} + \lVert y'' \rVert_1
  = \frac{\lVert x \rVert_1}{4} + \left\lceil\frac{\lVert x \rVert_1}{2}\right\rceil 
  \le \frac{\lVert x \rVert_1}{2} + \frac{\lVert x \rVert_1}{4} + \frac 1 2
  \le \frac 3 4 \lVert x \rVert_1 + \frac 1 {12} \lVert x \rVert_1
  \le \frac 5 6 \lVert x \rVert_1 .
\end{equation*}
Finally, $z_i = z'_i + z''_i \le y'_i + y''_i = x_i$ and
\begin{multline*}
  \left\lVert A \left(z - \frac x 2\right) \right\rVert_\infty
= \left\lVert A \left((z' + z'') - \frac{y' + y''} 2\right) \right\rVert_\infty \\
  \le \left\lVert A \left(z' - \frac{y'}{2}\right) \right\rVert_\infty + \left\lVert A \left(z'' - \frac{y''}{2}\right) \right\rVert_\infty
\le 2 \cdot \mathrm{herdisc}(A) . \qedhere
\end{multline*}
  \end{proof}
  \paragraph*{Computing the dynamic table using convolution.}
  In the following we explain the details on how to reduce the computation of the entries
  of the dynamic table to a $1$-dimensional convolution. We first need to handle
  that $2^{i-1-k}b$ might not be integral. Let $b^0 = \lfloor 2^{i-1-k}b \rfloor$ denote the
     vector rounded down in every component. Then $D_{i-1}$ is completely covered by
     the points with $\ell_\infty$-distance $4H+2$ from $b^0$. Likewise, $D_{i}$ is covered by the points with
     distance $4H+2$ from $2b^0$.

%
     We project a vector $b' \in D_{i-1}$ to
     \begin{equation}
       f_{i-1}(b') = \sum_{j=1}^{m} (16H + 11)^{j-1} \underbrace{(4H + 3 + b'_j - b^0)}_{\in \{1,\dotsc,8H + 5\}} . \label{sum-convolution}
     \end{equation}
     Notice that $16H + 11$ is always bigger than the sum of two values 
     of the form $4H + 3 + b'_j - b^0$.
     We define $f_{i}(b')$ for all $b'\in D_{i}$ in the same way, except we substitute $b^0$ for $2b^0$.
     For all $a,a'\in D_{i-1}, b'\in D_i$, it
     holds that $f_{i-1}(a) + f_{i-1}(a') = f_{i}(b')$, if and only if $a + a' = b' - (4H + 3,\dotsc,4H + 3)^T$:
     \paragraph*{Implication~$\Rightarrow$.}
       Let $f_{i-1}(a) + f_{i-1}(a') = f_{i}(b')$. Then, in particular,
       \begin{equation*}
  f_{i-1}(a) + f_{i-1}(a') \equiv f_{i}(b') \mod 16H + 11 
       \end{equation*}
       Since all but the first element of the sum (\ref{sum-convolution}) are multiples
       of $16H + 11$, i.e., they are equal $0$ modulo $16H + 11$, we can omit them
       in the equation. Hence,
       \begin{equation*}
  (4H + 3 + a_1 - b^0_1) + (4H + 3 + a'_1 - b^0_1)
  \equiv (4H + 3 + b'_1 - 2 b^0_1) \mod 16H + 11 .
       \end{equation*}
       We even have equality (without modulo) here, because both sides are smaller than
       $16m\Delta + 11$. Simplifying the equation gives
       $a_1  + a'_1 = b'_1 - (4H + 3)$.
       Now consider again the equation $f_{i-1}(a) + f_{i-1}(a') = f_{i}(b')$. In the sums
       leave out the first element. The equation still holds, since by the elaboration above this
       changes the left and right hand-side by the same value. We can now repeat the same argument
       to obtain $a_2  + a'_2 = b'_2 - (4H + 3)$ and the same for all other dimensions.
     \paragraph*{Implication~$\Leftarrow$.}
       Let $a + a' = b' - (4H + 3,\dotsc,4H + 3)^T$. Then for every $j$,
\begin{equation*}
  (4H + 3 + a_j - b^0_j) + (4H + 3 + a'_j - b^0_j)
  = 4H + 3 + b'_j - 2 b^0_j .
\end{equation*}
       It directly follows that $f_{i-1}(a) + f_{i-1}(a') = f_{i}(b')$.

     This means when we write the value of each $b''\in D_{i-1}$ to $r_j$ and $s_j$, where
     $j = f_{i-1}(b'')$ and every entry not used is set to $-\infty$,
     the correct solutions will be in $t$.
     More precisely, we can read the result for some $b'\in D_i$ at $t_j$ where
     $j = f_i(b' + (4H + 3,\dotsc,4H + 3)^T)$.

\end{document}

%% file: figure1.tex
\begin{tikzpicture}
  \draw (0, -0.5) rectangle (7, 5.5);
  \node[name=b, cross] at (6, 4.5) {};
  \draw (b) node[xshift=8, yshift=4] {$b$};
  \node[name=z, cross] at (1, 0.5) {};
  \draw (z) node[xshift=8, yshift=4] {$0$};
  \draw[thick, ->] (z) -- (0.6, 1.5);%
  \draw[thick, ->] (0.6, 1.5) -- (0.6, 2);%
  \draw[thick, ->] (0.6, 2) -- (0.8, 2.6);
  \draw[thick, ->] (0.8, 2.6) -- (1.4, 3);%
  \draw[thick, ->] (1.4, 3) -- (1, 3.5);%
  \draw[thick, ->] (1, 3.5) -- (1, 4.5);%
  \draw[thick, ->] (1, 4.5) -- (0.4, 3.9);%
  \draw[thick, ->] (0.4, 3.9) -- (1.2, 4);%
  \draw[thick, ->] (1.2, 4) -- (1.8, 4);%
  \draw[thick, ->] (1.8, 4) -- (2.5, 3.7);
  \draw[thick, ->] (2.5, 3.7) -- (3.2, 4.3);%
  \draw[thick, ->] (3.2, 4.3) -- (4.2, 4.3);%
  \draw[thick, ->] (4.2, 4.3) -- (4.8, 4);%
  \draw[thick, ->] (4.8, 4) -- (5.3, 3.5);%
  \draw[thick, ->] (5.3, 3.5) -- (5.7, 3.8);%
  \draw[thick, ->] (5.7, 3.8) -- (b);%
\end{tikzpicture}
\begin{tikzpicture}
  \draw (0, -0.5) rectangle (7, 5.5);
  \node[name=b, cross] at (6, 4.5) {};
  \draw (b) node[xshift=8, yshift=4] {$b$};
  \node[name=b2, cross] at (3.5, 2.5) {};
  \draw (b2) node[xshift=-9, yshift=4] {$\frac b 2$};
  \draw[thick, dashed] (2.8, 1.8) rectangle (4.2, 3.2);
  \node[name=z, cross] at (1, 0.5) {};
  \draw (z) node[xshift=8, yshift=-4] {$0$};
  \draw[thick, ->] (z) -- (1.6, 0.9);
  \draw[thick, ->] (1.6, 0.9) -- (2.3, 1.5);
  \draw[thick, ->] (2.3, 1.5) -- (2.3, 2);
  \draw[thick, ->] (2.3, 2) -- (2.8, 1.5);
  \draw[thick, ->] (2.8, 1.5) -- (3.1, 2.2);
  \draw[thick, ->] (3.1, 2.2) -- (3.9, 2.3);
  \draw[thick, ->] (3.9, 2.3) -- (3.3, 1.7);
  \draw[thick, ->] (3.3, 1.7) -- (3.3, 2.7);
  \draw[thick, ->] (3.3, 2.7) -- (3.7, 3);
  \draw[thick, ->] (3.7, 3) -- (4.3, 2.7);
  \draw[thick, ->] (4.3, 2.7) -- (3.9, 3.2);
  \draw[thick, ->] (3.9, 3.2) -- (4.5, 3.2);
  \draw[thick, ->] (4.5, 3.2) -- (4.7, 3.8);
  \draw[thick, ->] (4.7, 3.8) -- (5.7, 3.8);
  \draw[thick, ->] (5.7, 3.8) -- (5.3, 4.8);
  \draw[thick, ->] (5.3, 4.8) -- (b);

  \draw[thick, dashed] (0.3, -0.2) -- (0.3, 1.2) -- (5.3, 5.2)
     -- (6.7, 5.2) -- (6.7, 3.8) -- (1.7, -0.2) -- (0.3, -0.2);
\end{tikzpicture}

%% file: main-plain.bbl
\begin{thebibliography}{10}

\bibitem{DBLP:conf/soda/AbboudBHS19}
Amir Abboud, Karl Bringmann, Danny Hermelin, and Dvir Shabtay.
\newblock {SETH}-based lower bounds for subset sum and bicriteria path.
\newblock In {\em Proceedings of {SODA}}, pages 41--57, 2019.

\bibitem{agarwal1993efficient}
Pankaj~K Agarwal, Micha Sharir, and Sivan Toledo.
\newblock An efficient multi-dimensional searching technique and its
  applications.
\newblock In {\em Technical Report CS-1993-20}, 1993.

\bibitem{axiotis2018capacitated}
Kyriakos Axiotis and Christos Tzamos.
\newblock Capacitated dynamic programming: Faster knapsack and graph
  algorithms.
\newblock In {\em Proceedings of {ICALP}}, pages 19:1--19:13, 2019.

\bibitem{DBLP:conf/soda/BackursIS17}
Arturs Backurs, Piotr Indyk, and Ludwig Schmidt.
\newblock Better approximations for tree sparsity in nearly-linear time.
\newblock In {\em Proceedings of {SODA}}, pages 2215--2229, 2017.

\bibitem{DBLP:conf/focs/Bansal10}
Nikhil Bansal.
\newblock Constructive algorithms for discrepancy minimization.
\newblock In {\em Proceedings of {FOCS}}, pages 3--10, 2010.

\bibitem{DBLP:journals/dam/BeckF81}
J{\'{o}}zsef Beck and Tibor Fiala.
\newblock "integer-making" theorems.
\newblock {\em Discrete Applied Mathematics}, 3(1):1--8, 1981.

\bibitem{berndt2022load}
Sebastian Berndt, Max~A Deppert, Klaus Jansen, and Lars Rohwedder.
\newblock Load balancing: The long road from theory to practice.
\newblock In {\em Proceedings of the {ALENEX}}, pages 104--116, 2022.

\bibitem{DBLP:journals/algorithmica/BremnerCDEHILPT14}
David Bremner, Timothy~M. Chan, Erik~D. Demaine, Jeff Erickson, Ferran Hurtado,
  John Iacono, Stefan Langerman, Mihai Patrascu, and Perouz Taslakian.
\newblock Necklaces, convolutions, and {X+Y}.
\newblock {\em Algorithmica}, 69(2):294--314, 2014.

\bibitem{DBLP:conf/soda/Bringmann17}
Karl Bringmann.
\newblock A near-linear pseudopolynomial time algorithm for subset sum.
\newblock In {\em Proceedings of {SODA}}, pages 1073--1084, 2017.

\bibitem{DBLP:journals/siamcomp/BronnimannCM99}
Herv{\'{e}} Br{\"{o}}nnimann, Bernard Chazelle, and Jir{\'{\i}} Matousek.
\newblock Product range spaces, sensitive sampling, and derandomization.
\newblock {\em {SIAM} Journal on Computing}, 28(5):1552--1575, 1999.

\bibitem{DBLP:journals/talg/CHAN18}
Timothy~M. Chan.
\newblock Improved deterministic algorithms for linear programming in low
  dimensions.
\newblock {\em {ACM} Transactions on Algorithms}, 14(3):30:1--30:10, 2018.

\bibitem{DBLP:journals/jcss/ChanH22}
Timothy~M. Chan and Qizheng He.
\newblock More on change-making and related problems.
\newblock {\em Journal of Computer and System Sciences}, 124:159--169, 2022.

\bibitem{DBLP:conf/stoc/ChanL15}
Timothy~M. Chan and Moshe Lewenstein.
\newblock Clustered integer 3{SUM} via additive combinatorics.
\newblock In {\em Proceedings of {STOC}}, pages 31--40, 2015.

\bibitem{DBLP:journals/jal/ChazelleM96}
Bernard Chazelle and Jir{\'{\i}} Matousek.
\newblock On linear-time deterministic algorithms for optimization problems in
  fixed dimension.
\newblock {\em Journal of Algorithms}, 21(3):579--597, 1996.

\bibitem{DBLP:journals/ipl/Clarkson86}
Kenneth~L. Clarkson.
\newblock Linear programming in ${O}(n \times 3^{d^2})$ time.
\newblock {\em Information Processing Letters}, 22(1):21--24, 1986.

\bibitem{DBLP:journals/jacm/Clarkson95}
Kenneth~L. Clarkson.
\newblock Las vegas algorithms for linear and integer programming when the
  dimension is small.
\newblock {\em Journal of the {ACM}}, 42(2):488--499, 1995.

\bibitem{DBLP:conf/icalp/CyganMWW17}
Marek Cygan, Marcin Mucha, Karol Wegrzycki, and Michal Wlodarczyk.
\newblock On problems equivalent to (min, +)-convolution.
\newblock In {\em Proceedings of {ICALP}}, pages 22:1--22:15, 2017.

\bibitem{DBLP:journals/algorithmica/Dadush14}
Daniel Dadush.
\newblock A randomized sieving algorithm for approximate integer programming.
\newblock {\em Algorithmica}, 70(2):208--244, 2014.

\bibitem{dadush2012integer}
Daniel~Nicolas Dadush.
\newblock {\em Integer programming, lattice algorithms, and deterministic
  volume estimation}.
\newblock Georgia Institute of Technology, 2012.

\bibitem{DBLP:journals/siamcomp/Dyer86}
Martin~E. Dyer.
\newblock On a multidimensional search technique and its application to the
  euclidean one-centre problem.
\newblock {\em {SIAM} Journal on Computing}, 15(3):725--738, 1986.

\bibitem{DBLP:journals/mp/DyerF89}
Martin~E. Dyer and Alan~M. Frieze.
\newblock A randomized algorithm for fixed-dimensional linear programming.
\newblock {\em Mathematical Programming}, 44(1-3):203--212, 1989.

\bibitem{DBLP:journals/talg/EisenbrandW20}
Friedrich Eisenbrand and Robert Weismantel.
\newblock Proximity results and faster algorithms for integer programming using
  the {S}teinitz lemma.
\newblock {\em {ACM} Transactions on Algorithms}, 16(1):5:1--5:14, 2020.

\bibitem{DBLP:conf/esa/FominPR018}
Fedor~V. Fomin, Fahad Panolan, M.~S. Ramanujan, and Saket Saurabh.
\newblock On the optimality of pseudo-polynomial algorithms for integer
  programming.
\newblock In {\em Proceedings of {ESA} 2018}, pages 31:1--31:13, 2018.

\bibitem{DBLP:journals/mor/JansenKV20}
Klaus Jansen, Kim{-}Manuel Klein, and Jos{\'{e}} Verschae.
\newblock Closing the gap for makespan scheduling via sparsification
  techniques.
\newblock {\em Mathematics of Operations Research}, 45(4):1371--1392, 2020.

\bibitem{DBLP:conf/innovations/JansenR19}
Klaus Jansen and Lars Rohwedder.
\newblock On integer programming and convolution.
\newblock In {\em Proceedings of {ITCS}}, pages 43:1--43:17, 2019.

\bibitem{DBLP:conf/stoc/Kalai92}
Gil Kalai.
\newblock A subexponential randomized simplex algorithm (extended abstract).
\newblock In S.~Rao Kosaraju, Mike Fellows, Avi Wigderson, and John~A. Ellis,
  editors, {\em Proceedings of {STOC}}, pages 475--482, 1992.

\bibitem{DBLP:journals/mor/Kannan87}
Ravi Kannan.
\newblock Minkowski's convex body theorem and integer programming.
\newblock {\em Mathematics of Operations Research}, 12(3):415--440, 1987.

\bibitem{DBLP:books/daglib/0010031}
Hans Kellerer, Ulrich Pferschy, and David Pisinger.
\newblock {\em Knapsack problems}.
\newblock Springer, 2004.

\bibitem{klein2022fine}
Kim-Manuel Klein.
\newblock On the fine-grained complexity of the unbounded subsetsum and the
  {F}robenius problem.
\newblock In {\em Proceedings of {SODA}}, pages 3567--3582, 2022.

\bibitem{DBLP:conf/stacs/KnopPW19}
Dusan Knop, Michal Pilipczuk, and Marcin Wrochna.
\newblock Tight complexity lower bounds for integer linear programming with few
  constraints.
\newblock In {\em Proceedings of {STACS}}, pages 44:1--44:15, 2019.

\bibitem{DBLP:conf/icalp/KunnemannPS17}
Marvin K{\"{u}}nnemann, Ramamohan Paturi, and Stefan Schneider.
\newblock On the fine-grained complexity of one-dimensional dynamic
  programming.
\newblock In {\em Proceedings of {ICALP}}, pages 21:1--21:15, 2017.

\bibitem{DBLP:conf/isit/LaberRC14}
Eduardo~Sany Laber, Wilfredo~Bardales Roncalla, and Ferdinando Cicalese.
\newblock On lower bounds for the maximum consecutive subsums problem and the
  (min, +)-convolution.
\newblock In {\em Proceedings of {ISIT}}, pages 1807--1811, 2014.

\bibitem{DBLP:journals/mor/Lenstra83}
Hendrik~W. Lenstra.
\newblock Integer programming with a fixed number of variables.
\newblock {\em Mathematics of Operations Research}, 8(4):538--548, 1983.

\bibitem{DBLP:journals/ejc/LovaszSV86}
L{\'{a}}szl{\'{o}} Lov{\'{a}}sz, Joel Spencer, and Katalin Vesztergombi.
\newblock Discrepancy of set-systems and matrices.
\newblock {\em European Journal of Combinatorics}, 7(2):151--160, 1986.

\bibitem{DBLP:journals/siamcomp/LovettM15}
Shachar Lovett and Raghu Meka.
\newblock Constructive discrepancy minimization by walking on the edges.
\newblock {\em {SIAM} Journal on Computing}, 44(5):1573--1582, 2015.

\bibitem{DBLP:journals/algorithmica/MatousekSW96}
Jir{\'{\i}} Matousek, Micha Sharir, and Emo Welzl.
\newblock A subexponential bound for linear programming.
\newblock {\em Algorithmica}, 16(4/5):498--516, 1996.

\bibitem{DBLP:journals/jacm/Megiddo84}
Nimrod Megiddo.
\newblock Linear programming in linear time when the dimension is fixed.
\newblock {\em Journal of the {ACM}}, 31(1):114--127, 1984.

\bibitem{DBLP:journals/jacm/Papadimitriou81}
Christos~H. Papadimitriou.
\newblock On the complexity of integer programming.
\newblock {\em Journal of the {ACM}}, 28(4):765--768, 1981.

\bibitem{DBLP:journals/dcg/Seidel91}
Raimund Seidel.
\newblock Small-dimensional linear programming and convex hulls made easy.
\newblock {\em Discrete \& Computational Geometry}, 6:423--434, 1991.

\bibitem{sevast-steinitz}
Sergey~V. Sevastyanov.
\newblock Approximate solution of some problems in scheduling theory.
\newblock {\em Metody Diskretnogo Analiza}, 32:66--75, 1978.
\newblock in Russian.

\bibitem{spencer1985six}
Joel Spencer.
\newblock Six standard deviations suffice.
\newblock {\em Transactions of the American mathematical society},
  289(2):679--706, 1985.

\bibitem{SteinitzLemma}
Ernst Steinitz.
\newblock Bedingt konvergente reihen und konvexe systeme.
\newblock {\em Journal f\"ur die reine und angewandte Mathematik},
  143:128--176, 1913.
\newblock in German.

\end{thebibliography}
